\documentclass[11pt]{article}

\usepackage{graphicx,amsmath,amsfonts,amssymb,dcolumn,amsthm,euscript}

\newtheorem{thma}{Theorem}[section]
\newtheorem{defi}{Definition}[section]

\begin{document}

\title{\textbf{Non-perturbative treatment of the linear covariant gauges by taking into account the Gribov copies}}
\author{\textbf{M.~A.~L.~Capri$^1$}\thanks{caprimarcio@gmail.com}\ , \textbf{A.~D.~Pereira$^2$}\thanks{aduarte@if.uff.br}\ , \textbf{R.~F.~Sobreiro$^2$}\thanks{sobreiro@if.uff.br}\ , \textbf{S.~P.~Sorella$^1$}\thanks{silvio.sorella@gmail.com}\\\\
\textit{{\small $^1$UERJ $-$ Universidade do Estado do Rio de Janeiro,}}\\
\textit{{\small Departamento de F\'isica Te\'orica, Rua S\~ao Francisco Xavier 524,}}\\
\textit{{\small 20550-013, Maracan\~a, Rio de Janeiro, Brasil}}\and
\textit{{\small $^2$UFF $-$ Universidade Federal Fluminense,}}\\
\textit{{\small Instituto de F\'{\i}sica, Campus da Praia Vermelha,}}\\
\textit{{\small Avenida General Milton Tavares de Souza s/n, 24210-346,}}\\
\textit{{\small Niter\'oi, RJ, Brasil.}}}
\date{}
\maketitle

\begin{abstract}
In this paper, a proposal for the restriction of the Euclidean functional integral to a region free of infinitesimal Gribov copies in linear covariant gauges is discussed. An effective action, akin to the Gribov-Zwanziger action of the Landau gauge, is obtained which implements the aforementioned restriction.  Although originally non-local, this action can be cast in local form  by  introducing auxiliary fields. As in the case of the Landau gauge, dimension two condensates are generated at the quantum level, giving rise to a refinement of the action which is employed to obtain the tree-level gluon propagator in linear covariant gauges. A comparison of our results with those available from numerical lattice simulations  is also provided.
\end{abstract}

\section{Introduction} \label{Intro}

Analytical approaches to  Quantum Chromodynamics (QCD) face many  problems due to the fact that at low energy scales, where many important physical phenomena as confinement  and chiral symmetry breaking take place, perturbation theory breaks down.  Moreover, a  complete non-perturbative framework is not yet at our disposal. On the other hand, lattice numerical methods have provided firm evidence for those non-perturbative effects. Let us underline that the interplay between numerical and analytical analysis has been proven to be very fruitful over the last decade.  
 
\noindent An analytical approach which is receiving growing attention is the  Gribov-Zwanziger set up. In this framework, a suitable modification of  the Faddeev-Popov path integral quantization procedure  is worked out and non-trivial non-perturbative results emerge. As first noticed by Gribov \cite{Gribov:1977wm}, after gauge fixing, one still has gauge field configurations which satisfy the gauge condition and are connected to each other through a gauge transformation: the so-called Gribov copies.  This implies that a residual redundancy is still present in the Faddeev-Popov quantization formula.   Although this fact was originally pointed out  in  Landau and Coulomb gauges \cite{Gribov:1977wm}, it was proven by Singer \cite{Singer:1978dk} that this is not a peculiarity of a particular gauge choice, being  a general feature of the gauge fixing procedure. 

\noindent In \cite{Gribov:1977wm}\footnote{See  \cite{Sobreiro:2005ec,Vandersickel:2012tz} for a pedagogical introduction to the Gribov problem.} a partial solution to this problem was proposed in the Landau gauge, $\partial_\mu A^a_\mu=0$, for which the corresponding Faddeev-Popov  operator $\mathcal{M}^{ab}$,

\begin{equation}
\mathcal{M}^{ab}=-\delta^{ab}\partial^2+gf^{abc}A^{c}_{\mu}\partial_{\mu},\,\,\,\, \mathrm{with}\,\,\,\, \partial_{\mu}A^{a}_{\mu}=0\,,      
\label{intro0}
\end{equation}

\noindent is hermitian \cite{Sobreiro:2005ec,Vandersickel:2012tz}, so that its eigenvalues are real. The original proposal by Gribov \cite{Gribov:1977wm} in order to take into account the existence of gauge copies was that of restricting the domain of integration in the functional Euclidean integral to a certain region $\Omega$ in field space defined as the set of all field configurations fulfilling the Landau gauge and for which the Faddeev-Popov operator is strictly positive, namely 
\begin{equation}
\Omega = \{ \; A^a_\mu,  \;  \partial_\mu A^a_\mu =0,  \;    {\mathcal M}^{ab} > 0 \; \}    \;. \label{om}
\end{equation}
The region $\Omega$ is called the Gribov region.   It enjoys several properties \cite{Dell'Antonio:1991xt}: {\it 1)}  it is a bounded region in field space, {\it 2)} it is convex, {\it 3)}  all gauge orbits cross $\Omega$ at least once. In particular, this last property ensures that a gauge field configuration located outside of the Gribov region  is a copy of a configuration located inside $\Omega$. This important result gives a  well defined support\footnote{Nowadays, it is known that the Gribov region $\Omega$ is not completely free from Gribov copies, see for example \cite{Vandersickel:2012tz} and references therein.  There are still additional copies inside $\Omega$. A smaller region, known as the fundamental modular region (FMR) has been identified which is completely free from Gribov copies. Though, our present knowledge of the FMR is not as good as our knowledge of the Gribov region $\Omega$, which will be the main subject of the present investigation.} to Gribov's original proposal of restricting the functional integral to the region $\Omega$. 

\noindent Therefore, according to Gribov's proposal \cite{Gribov:1977wm}, for the partition function of quantized Yang-Mills action in the Landau gauge we write 
\begin{equation}
\EuScript{Z}=\int_{\Omega} \left[\EuScript{D}\mathbf{\Phi}\right]\mathrm{e}^{-\left( S_{\mathrm{YM}}+S_{\mathrm{gf}}\right) }\,,
\label{intro1}
\end{equation}

\noindent with $\left( S_{\mathrm{YM}}+S_{\mathrm{gf}} \right) $ given by eq.(\ref{A2}) with $\alpha=0$. Later on, Zwanziger  \cite{Zwanziger:1989mf} was able to show that the restriction of the domain of integration to the region $\Omega$ can be effectively implemented by adding to the starting action an additional non-local term $H_{L}(A)$, known as the horizon function. More precisely 
\begin{equation}
\int_{\Omega} \left[\EuScript{D}\mathbf{\Phi}\right]\mathrm{e}^{-\left( S_{\mathrm{YM}}+S_{\mathrm{gf}}\right) } = \int \left[\EuScript{D}\mathbf{\Phi}\right]\mathrm{e}^{-    {S}^{L}_{\mathrm{GZ}}  }   \;, 
\label{gz1}
\end{equation}
where 
\begin{equation}
{S}^{L}_{\mathrm{GZ}}=   S_{\mathrm{YM}}+S_{\mathrm{gf}}    +\gamma^4H_{L}(A) - 4V\gamma^4(N^2-1)\,, 
\label{intro2}
\end{equation}
and 
\begin{equation}
H_{L}(A)=g^2\int d^4x~f^{abc}A^{b}_{\mu}\left[\mathcal{M}^{-1}\right]^{ad}f^{dec}A^{e}_{\mu}\,, 
\label{intro3}
\end{equation}
with $\left[\mathcal{M}^{-1}\right]$ denoting the inverse of the Faddeev-Popov operator, eq.\eqref{intro0}. 
The mass parameter $\gamma$ appearing in expression \eqref{intro2}  is known as the Gribov parameter. It  is not a free parameter, being  determined in a self-consistent way by the  gap equation \cite{Zwanziger:1989mf}
\begin{equation}
\langle H_{L} \rangle = 4V(N^2-1)\,.
\label{intro4}
\end{equation}
where the vacuum expectation values $\langle H_{L} \rangle$ has to be evaluated with the measure defined in eq.\eqref{gz1}. 
\noindent  

\noindent The action  in eq.(\ref{intro2}) is called the Gribov-Zwanziger (GZ) action. It is a non-local action due to the presence of the horizon function. However, it can  be cast in local form \cite{Vandersickel:2012tz} by the introduction of a pair of commuting fields $(\bar{\varphi}^{ab}_{\mu},\varphi^{ab}_{\mu})$ and a pair of anti-commuting ones $(\bar{\omega}^{ab}_{\mu},\omega^{ab}_{\mu})$, namely 

\begin{eqnarray}
S^{L}_{\mathrm{GZ}} &=& S_{\mathrm{YM} } +S_{\mathrm{gf}}-\int d^4x\left(\bar{\varphi}^{ac}_{\mu}\mathcal{M}^{ab}\varphi^{bc}_{\mu}-\bar{\omega}^{ac}_{\mu}\mathcal{M}^{ab}\omega^{bc}_{\mu} + gf^{adb}(\partial_{\nu}\bar{\omega}^{ac}_{\mu})(D^{de}_{\nu}c^{e})\varphi^{bc}_{\mu}\right) \nonumber \\
&+& \gamma^2\int d^4x~gf^{abc}A^{a}_{\mu}(\varphi+\bar{\varphi})^{bc}_{\mu}-4\gamma^4V(N^2-1)\,.
\label{intro5}
\end{eqnarray}
Since the auxiliary fields $(\bar{\varphi}^{ab}_{\mu},\varphi^{ab}_{\mu})$ and $(\bar{\omega}^{ab}_{\mu},\omega^{ab}_{\mu})$ appear quadratically in expression \eqref{intro5}, they can be easily integrated out, giving back the non-local expression \eqref{intro2}. Remarkably, the local action (\ref{intro5}) turns out to be renormalizable to all orders \cite{Vandersickel:2012tz,Zwanziger:1989mf}, while implementing the restriction to the Gribov region $\Omega$ in the Landau gauge. 

\noindent The Gribov-Zwanziger action in the Landau gauge is plagued by dynamical non-perturbative instabilities and the generation of dimension two condensates is favoured \cite{Dudal:2007cw,Dudal:2008sp,Dudal:2011gd}. These condensates can be taken into account by further modifying the original  action, giving rise to the so-called refined Gribov-Zwanziger (RGZ) action  \cite{Dudal:2007cw,Dudal:2008sp,Dudal:2011gd}, given by 

\begin{equation}
S^{L}_{\mathrm{RGZ}} = S^{L}_{\mathrm{GZ}} + \frac{m^2}{2}\int d^4x~A^{a}_{\mu}A^{a}_{\mu} - M^2\int d^4x\left(\bar{\varphi}^{ab}_{\mu}\varphi^{ab}_{\mu}-\bar{\omega}^{ab}_{\mu}\omega^{ab}_{\mu}\right)\,.
\label{intro6}
\end{equation}
As much as the Gribov parameter $\gamma^2$, the parameters $(m,M)$ are not free, being generated dynamically in a self-consistent way, see  \cite{Dudal:2008sp,Dudal:2011gd}. 

\noindent The RGZ action describes gluon and ghost propagators in very good agreement with the most recent lattice results \cite{Cucchieri:2007rg,Cucchieri:2008fc,Cucchieri:2011ig}. In particular, for the the tree level gluon propagator one obtains 

\begin{equation}
\langle A^{a}_{\mu}(k)A^{b}_{\nu}(-k)\rangle_{L} = \delta^{ab}\left(\delta_{\mu\nu}-\frac{k_{\mu}k_{\nu}}{k^2}\right)\frac{k^2+M^2}{(k^2+m^2)(k^2+M^2)+2Ng^2\gamma^4}\,.
\label{intro7}
\end{equation}

\noindent Also, the RGZ framework gives results for glueball masses which compare well with numerical results, see \cite{Dudal:2010cd,Dudal:2013wja} and references therein.

\noindent It is worth underlining here that the construction of the  GZ and RGZ actions in the Landau gauge follows  directly from  the excellent understanding of the properties of the Gribov region $\Omega$ of the Landau gauge which has been achieved so far, as mentioned before. Willing to address the Gribov issue in other gauges, the first and necessary step towards the construction of a suitable non-perturbative action is the study of the corresponding Gribov region, whose properties are encoded in  the structure of the Faddeev-Popov operator associated to the chosen gauge. This is, for example, the way in which the Gribov issue has been addressed in the maximal Abelian gauge  \cite{Capri:2005tj,Dudal:2006ib,Capri:2006cz,Capri:2008vk,Capri:2010an}, where it has been possible to characterize a few properties of the corresponding Gribov region, paving the construction of the analogue of the RGZ action in this gauge.  Let us mention  that, besides the Landau and maximal Abelian gauges, a successful characterization of the Gribov region and of the ensuing GZ and RGZ actions has been worked out also in the non-covariant Coulomb gauge  \cite{Zwanziger:2002sh,Zwanziger:2003de,Greensite:2004ke,Greensite:2004ke,Zwanziger:2004np,Guimaraes:2015bra}. All these three gauges share the very important feature that the corresponding Faddeev-Popov operators are hermitian. Thus, they have real eingevalues, a property which makes the definition of the corresponding Gribov regions very clear and transparent. In addition, all these gauges can be defined through the minimization of suitable  auxiliary functionals. For example, the Landau gauge, $\partial_\mu A^a_\mu=0$, can be obtained as the stationary condition with respect to the gauge transformations of the auxiliary functional $\int d^4x A^a_\mu A^a_\mu$. Moreover, the second variation of $\int d^4x A^a_\mu A^a_\mu$ with respect to the gauge transformations yields precisely  the Faddeev-Popov operator of the Landau gauge, see \cite{Vandersickel:2012tz}. As a consequence, the Gribov region $\Omega$ of the Landau gauge, eq.\eqref{om}, can be defined as the set of all relative minima of the auxiliary functional  $\int d^4x A^a_\mu A^a_\mu$ in field space. A similar feature holds in both maximal Abelian and Coulomb gauges. This important property has deep mathematical consequences and allows for a lattice formulation of these gauges,  whose predictions can be thus tested from the numerical side.  

\noindent In general, addressing the issue of the Gribov copies in an arbitrary gauge is a  very difficult task, as many properties of the Landau, Coulomb and maximal Abelian gauges are lost. This is, for instance, the case of  the  linear covariant gauges which are the object of the present investigation. Here, unlike the Landau, Coulomb and maximal Abelian gauges, the Faddeev-Popov operator lacks hermiticity and, moreover, no auxiliary functional is at our disposal, making the treatment of the Gribov copies highly non-trivial. Despite these difficulties, a first attempt to address the Gribov problem in the linear covariant gauges was discussed in \cite{Sobreiro:2005vn}, under the assumptions that the gauge parameter $\alpha$ present in these  gauges was considered to be infinitesimal, {\it i.e.} $\alpha \ll 1$. This hypothesis allowed the authors \cite{Sobreiro:2005vn}  to introduce  a region in field space which is free from  infinitesimal Gribov copies.  More precisely, according to \cite{Sobreiro:2005vn}, this region  is defined by requiring that the transverse component of the gauge field\footnote{We remind  that the class of the linear covariant gauges is defined by the condition: $\partial_\mu A^a_\mu -\alpha b^a=0$, where $\alpha$ is a gauge parameter. Unlike the Landau gauge, the gauge field $A^a_\mu$ has now both transverse and longitudinal components, {\it i.e.} $A^a_\mu = A^{aT}_\mu + A^{aL}_\mu$. } $A^{aT}_{\mu}=(\delta_{\mu\nu}-\partial_{\mu}\partial_{\nu}/\partial^2)A^{a}_{\nu}$ belongs to the Gribov region  $\Omega$ of the Landau gauge, eq.\eqref{om}, and that the longitudinal component  $A^{aL}_{\mu}=(\partial_{\mu}\partial_{\nu}/\partial^2)A^a_{\nu}$ remains free. 

\noindent In this work we pursue the analysis of the linear covariant gauges  started in \cite{Sobreiro:2005vn}. In particular, we  identify a suitable region  ${\Omega}_{\mathrm{LCG}}$ in field space which turns out to be free from infinitesimal copies, without relying on the assumption $\alpha \ll 1$ made in  \cite{Sobreiro:2005vn}. More precisely, as we shall see in the next section, we will be able to eliminate all infinitesimal copies related to zero modes of the Faddeev-Popov operator which are smooth functions of the gauge parameter $\alpha$. Also, the restriction of the  domain of integration in the Euclidean functional integral  is effectively implemented through the introduction of the correspondent non-local horizon function which can be localized by means of a set of auxiliary fields. Though, the  localization procedure in the linear covariant gauges displays  new features with respect to the Landau gauge. Therefore, a local Gribov-Zwanziger action which eliminates infinitesimal Gribov copies in linear covariant gauges is established. The gap equation which determines the Gribov parameter $\gamma$ is analyzed at one-loop order, with the important physical outcome that no gauge parameter dependence emerges for $\gamma$.  As in the case of the Landau, maximal Abelian and Coulomb gauges, this action is plagued by further non-perturbative dynamical effects related to dimension two condensates, resulting in a refinement of the starting Gribov-Zwanziger action. The tree-level gluon propagator stemming from the refined Gribov-Zwanziger action in linear covariant gauges is evaluated and compared with the  most recent numerical lattice lattice data.  

\noindent The work is organized as follows. In Sect.\ref{gribovreg} we elaborate on the definition of the region ${\Omega}_{\mathrm{LCG}}$ which allows to eliminate the infinitesimal copies in the linear covariant gauges. In Sect.\ref{horizon}  we implement the restriction of the domain of integration in the functional integral to the region ${\Omega}_{\mathrm{LCG}}$ by introducing a suitable non-local horizon function.  In Sect.\ref{localization} we show how to localize the action defined in the previous section by introducing a set of  auxiliary fields. The resulting action is identified with the local Gribov-Zwanziger action in linear covariant gauges.  In Sect.\ref{brstbreaking}, the BRST soft breaking of the Gribov-Zwanziger action is briefly analyzed.  A consistency check is realized in Sect.\ref{oneloopgap}, where we show that the gap equation which determines the Gribov parameter $\gamma$ turns out to be independent from the gauge parameter at one-loop order. Further, the origin of dimension two condensates by dynamical effects is analyzed in Sect.\ref{condensatessect} as well as the refinement of the Gribov-Zwanziger action. In Sect.\ref{propagator}, we  present the tree-level gluon propagator and compare it with the available  results from numerical lattice simulations. Finally, we present our conclusions and perspectives.

\section{Proposal for a  Gribov region in linear covariant gauges} \label{gribovreg}

As discussed in Sect.\ref{Intro}, a partial resolution of the  Gribov problem in the Landau gauge requires the restriction of the domain of integration in the path integral to the so-called Gribov region $\Omega$, eq.\eqref{om}.  In the Landau gauge, since the Faddeev-Popov operator is hermitian, the definition of $\Omega$ as the region in field space  where the Faddeev-Popov operator is positive has a clear and well defined geometrical meaning, as already mentioned before.   This is not the case of linear covariant gauges, where the gauge condition reads

\begin{equation}
\partial_{\mu}A^{a}_{\mu}-\alpha b^{a}=0\,,
\label{lcg1}
\end{equation}

\noindent where $\alpha$ is an arbitrary gauge parameter and $b^a$ a Lagrange multiplier field. This gauge condition coincides with Landau condition when $\alpha=0$. However, for values of $\alpha$ different from zero, the gauge field acquires a longitudinal component, which is absent in the Landau gauge. If we perform an infinitesimal gauge transformation on eq.(\ref{lcg1}) and impose the gauge condition (\ref{lcg1}) on the gauge transformed field, we obtain

\begin{eqnarray}
\EuScript{M}^{ab}\xi^b & \equiv & -\partial_{\mu}D^{ab}_{\mu}\xi^b=0\,, \label{lcg2}  \\
\EuScript{M}^{ab} &=& -\partial_{\mu}D^{ab}_{\mu} = - \partial_\mu ( \delta^{ab}\partial_\mu  + g f^{acb} A^c_\mu )  = -\delta^{ab} \partial^2 -\alpha g f^{acb}b^c - gf^{acb}A^c_\mu \partial_\mu \;, 
\label{new1}
\end{eqnarray}
with $\xi^a$ being the infinitesimal parameter of the gauge transformation. The operator $\EuScript{M}^{ab}$ in eq.\eqref{new1} 
is the Faddeev-Popov operator of the linear covariant gauges and eq.\eqref{lcg2} is the corresponding zero mode equation giving rise to infinitesimal Gribov copies, namely 
\begin{equation}
{\hat A}^{a}_\mu = A^a_\mu + D^{ab}_\mu \xi^b \;, \qquad \partial_\mu {\hat A}^{a}_\mu = \partial_\mu A^a_\mu =  \alpha b^a \;.  \label{new2}   \; 
\end{equation}
Unlike the Faddeev-Popov operator of the  Landau gauge, the operator  $\EuScript{M}^{ab}$ in eq.\eqref{new1} is not hermitian. This fact makes the characterization of the analogue of the Gribov region in the linear covariant gauges highly non-trivial. We  highlight that the elimination of the Gribov copies using the approach followed by Gribov \cite{Gribov:1977wm} and by Zwanziger \cite{Vandersickel:2012tz,Zwanziger:1989mf} requires the use of a hermitian operator.  The understanding of the properties of the so-called Gribov region is crucial in order to deal with the problem of the gauge copies. As already mentioned, for a generic choice of the gauge condition, the properties of the corresponding Gribov region are encoded in the associated Faddeev-Popov operator and no systematic way of defining it is known.  

\noindent Nevertheless, an attempt to define a region free from infinitesimal copies in linear covariant gauges was done in \cite{Sobreiro:2005vn}. In this work, a region free from infinitesimal copies was identified, under the assumption that  the gauge parameter $\alpha$ was infinitesimal, {\it i.e.} $\alpha \ll 1$. The region introduced in \cite{Sobreiro:2005vn} is defined by demanding that  the transverse component of the gauge field, $A^{aT}_\mu= \left(\delta_{\mu\nu}-\frac{\partial_{\mu}\partial_{\nu}}{\partial^2}\right)A^{a}_{\nu}$,  belongs to the Gribov region $\Omega$ of Landau gauge, eq.\eqref{om}. Here, we extend this result for a finite value of $\alpha$. The extension relies on  the following theorem. 

\begin{thma}
If the transverse component of the gauge field $A^{Ta}_{\mu}=\left(\delta_{\mu\nu}-\frac{\partial_{\mu}\partial_{\nu}}{\partial^2}\right)A^{a}_{\nu}$ $\in$ $\Omega$, then the equation $\EuScript{M}^{ab}\xi^b=0$ has only the trivial solution $\xi^b=0$.
\label{thma}
\end{thma}

\begin{proof}
Since $A^{aT}_\mu$ $\in$ $\Omega$ by assumption, the  operator

\begin{equation}
\EuScript{M}^{Tab}\equiv -\delta^{ab}\partial^2+gf^{abc}A^{Tc}_{\mu}\partial_{\mu}\,,
\label{proof1}
\end{equation}

\noindent is positive definite and, therefore, is invertible. As a consequence, eq.(\ref{lcg2}) can be rewritten as

\begin{eqnarray}
\xi^a(x,\alpha) &=& -g\left[(\EuScript{M}^T)^{-1}\right]^{ad}f^{dbc}\partial_{\mu}(A^{Lc}_{\mu}\xi^{b}) \nonumber \\
&=& -g\alpha\left[(\EuScript{M}^T)^{-1}\right]^{ad}f^{dbc}\partial_{\mu}\left(\left(\frac{\partial_{\mu}}{\partial^2}b^{c}\right)\xi^{b}\right)\,,
\label{proof2}
\end{eqnarray}

\noindent where the gauge condition (\ref{lcg1}) was used. We consider here zero modes  $\xi(x,\alpha)$ which are smooth functions of the gauge parameter $\alpha$. This requirement is motivated by the physical consideration that the quantity $\xi(x, \alpha)$ corresponds to the parameter of an infinitesimal gauge transformation. On physical grounds, we expect thus a regular behaviour of $\xi(x,\alpha)$ as function of $\alpha$, \textit{i.e.} infinitesimal modifications on the value of $\alpha$ should not produce a drastic singular behaviour of $\xi$, a feature also supported by the important fact that an acceptable zero mode has to be a square-integrable function, {\it i.e.} $\int d^4x \;\xi^a \xi^a <\infty$. Also, the $\alpha$-dependence should be such that in the limit $\alpha \rightarrow 0$ we recover the zero-modes of the Landau gauge.  Therefore, we require smoothness of $\xi(x,\alpha)$ with respect to $\alpha$.  Thus, for a certain radius of convergence $R$, we can write the zero-mode $\xi(x,\alpha)$ as a Taylor expansion in $\alpha$,

\begin{equation}
\xi^a(x,\alpha)=\sum^{\infty}_{n=0}\alpha^n\xi^a_{n}(x)\,.
\label{proof3}
\end{equation}

\noindent For such radius of convergence, we can plug eq.(\ref{proof3}) into eq.(\ref{proof2}), which gives

\begin{equation}
\sum^{\infty}_{n=0}\alpha^n\xi^a_{n}(x)=-\sum^{\infty}_{n=0}g\alpha^{n+1}\left[(\EuScript{M}^T)^{-1}\right]^{ad}f^{dbc}\partial_{\mu}\left(\left(\frac{\partial_{\mu}}{\partial^2}b^{c}\right)\xi^{b}_{n}\right)\equiv \sum^{\infty}_{n=0}\alpha^{n+1}\phi^a_{n}\,.
\label{proof4}
\end{equation}

\noindent Since $\alpha$ is arbitrary, eq.(\ref{proof4}) should hold order by order in $\alpha$, which implies

\begin{equation}
\xi^a_0=0\,\,\,\, \Rightarrow \,\,\,\, \phi^a_0 = -\alpha\left[(\EuScript{M}^T)^{-1}\right]^{ad}f^{dbc}\partial_{\mu}\left(\left(\frac{\partial_{\mu}}{\partial^2}b^{c}\right)\xi^{b}_{0}\right)=0\,,
\label{proof5}
\end{equation}

\noindent at zeroth order. Therefore, 

\begin{equation}
\xi^a_1 = \alpha\phi^a_0=0\,,
\label{proof5a}
\end{equation}

\noindent and by recursion,

\begin{equation}
\xi^a_n = \alpha\phi^a_{n-1}=0,\,\,\forall n\,.
\label{proof6}
\end{equation}

\noindent Hence, the zero-mode $\xi(x,\alpha)$ must be identically zero within $R$. Due to the requirement of smoothness, {\it i.e.}  of differentiability and  continuity of $\xi$, the zero-mode must vanish everywhere.
\end{proof}

\noindent A comment is in order here. As underlined before, Theorem~\ref{thma} holds for zero-modes which have a Taylor expansion in powers of $\alpha$. Although this does not seem to be a very strong requirement, since we expect smooth functions of $\alpha$, \textit{i.e.} small perturbations on $\alpha$ should not result on abrupt changes on $\xi$, one could think about the possibility to have zero-modes which might eventually display a pathological behaviour, {\it i.e.} which could be singular for some values of the gauge parameter $\alpha$. For this reason, we shall refer to  the  zero-modes that admit a Taylor expansion as \textit{regular} zero-modes. 

\noindent Motivated by the previous Theorem, we introduce now the following Gribov region ${\Omega}_{\mathrm{LCG}}$ in the linear covariant gauges:  

\begin{defi}
The Gribov region ${\Omega}_{\mathrm{LCG}}$ in linear covariant gauges is given by

\begin{equation}
{\Omega}_{\mathrm{LCG}} = \left\{A^{a}_{\mu},\,\, \partial_{\mu}A^{a}_{\mu}-\alpha b^a =0,\,\,\EuScript{M}^{Tab} > 0 \right\}\,,    \; 
\label{defi1}
\end{equation}

\label{defi}
\end{defi}
where the operator $\EuScript{M}^{Tab}$ is given in eq.\eqref{proof1}

\noindent From the previous Theorem, it follows that the  region ${\Omega}_{\mathrm{LCG}}$ is free from  infinitesimal Gribov copies which are regular, {\it i.e.} smooth functions of the gauge parameter $\alpha$.   

\noindent Although the understanding of the Gribov issue in the linear covariant gauges cannot certainly be compared to that achieved in the Landau gauge, the possibility of introducing the region ${\Omega}_{\mathrm{LCG}}$ which is free from infinitesimal copies can be regarded a first important step in order to face this highly non-trivial problem.  Let us now focus on the construction of the horizon function in this class of  gauges

\section{The horizon function and the Gribov-Zwanziger action in linear covariant gauges} \label{horizon}

Definition~\ref{defi} provides a consistent candidate for the Gribov region in linear covariant gauges. Following the approach employed  by Gribov in \cite{Gribov:1977wm} and generalized by Zwanziger \cite{Zwanziger:1989mf}, we should restrict the path integral to the region ${\Omega}_{\mathrm{LCG}}$, \textit{i.e.}

\begin{equation}
\EuScript{Z} = \int_{{\Omega}_{\mathrm{LCG}}}  \left[\EuScript{D}\mathbf{\Phi}\right] \mathrm{e}^{-(S_{\mathrm{YM}}+S_\mathrm{gf})}\,,
\label{hor1}
\end{equation}

\noindent where $\Phi$ represents all fields of the theory, see Appendix.~\ref{conventions} for the conventions used. From equation \eqref{defi1},  one immediately sees that the region ${\Omega}_{\mathrm{LCG}}$ is defined by the positivity of the operator  $\EuScript{M}^T$ which contains only the transverse component of the gauge field, eq.\eqref{proof1}. In other words, $\EuScript{M}^T$ is nothing but the Faddeev-Popov operator of the Landau gauge.  As a consequence, the whole procedure performed by Gribov \cite{Gribov:1977wm} and Zwanziger \cite{Zwanziger:1989mf} in the case of the Landau gauge, can be repeated here, although one has to keep in mind that the restriction of the domain of integration to the region ${\Omega}_{\mathrm{LCG}}$ affects only the transverse component of the gauge field, while the longitudinal sector remains unmodified. Therefore, following \cite{Gribov:1977wm,Zwanziger:1989mf}, for the restriction to the region ${\Omega}_{\mathrm{LCG}}$ we write 

\begin{equation}
\int_{{\Omega}_{\mathrm{LCG}}}  \left[\EuScript{D}\mathbf{\Phi}\right] \mathrm{e}^{-(S_{\mathrm{YM}}+S_\mathrm{gf})} = \int  \left[\EuScript{D}\mathbf{\Phi}\right] \mathrm{e}^{-\tilde{S}_{\mathrm{GZ}}} \;, \label{new3}   \; 
\end{equation}
where the action  $\tilde{S}_{\mathrm{GZ}}$ is given by 
\begin{equation}
\tilde{S}_{\mathrm{GZ}} = S_{\mathrm{YM}} + S_{\mathrm{gf}} + \gamma^4H(A^T) - 4V\gamma^4(N^2-1)\,.
\label{hor3}
\end{equation}
The quantity $H(A^T)$ is the non-local horizon function which depends only on the transverse component of the gauge field $A^T$, namely 

\begin{equation}
H(A^T) =  g^2 \int d^4x~f^{adc}A^{Tc}_{\mu}[\left(\EuScript{M}^T\right)^{-1}]^{ab}f^{bde}A^{Te}_{\mu}\,.
\label{hor4}
\end{equation}

\noindent The parameter $\gamma$ in eq.\eqref{hor3} is the Gribov parameter. As in the case of the Landau gauge, it is  determined in a  self-consistent  way by the gap equation

\begin{equation}
\langle H(A^T) \rangle = 4V(N^2-1)\,,
\label{gapequation}
\end{equation}
where the vacuum expectation value $\langle H(A^T) \rangle$ has to be evaluated  now with the measure defined in eq.\eqref{new3}.

\noindent The effective action (\ref{hor3}) implements the restriction to the region ${\Omega}_{\mathrm{LCG}}$. Here, an important feature has to be pointed out. Formally, the horizon function (\ref{hor4}) is the same as in Landau gauge. However, in this case, although the longitudinal component of the gauge field does not enter the horizon function, it appears explicitly in the action $\tilde{S}_{\mathrm{GZ}}$. As we shall see, this property will give rise to several differences with respect to the Landau gauge.    Another important point to be underlined concerns  the vacuum term $4V\gamma^4(N^2-1)$ in expression \eqref{hor3}. This term is related to the spectrum of the operator $\EuScript{M}^{T}$ \cite{Vandersickel:2012tz} which does not depend on the gauge parameter $\alpha$. Therefore, at least at the level of the construction of the effective action $\tilde{S}_{\mathrm{GZ}}$, eq.(\ref{hor3}), the vacuum term is independent of $\alpha$. Of course, we should check out  if quantum corrections might eventually introduce some dependence from $\alpha$ in the vacuum term. This would require a lengthy analysis of the renormalizability properties of $\tilde{S}_{\mathrm{GZ}}$, which is under current  investigation \cite{Capri}.

\noindent From now on, we will refer to the action (\ref{hor3}) as the Gribov-Zwanziger action in linear covariant gauges. As it happens in the Landau gauge, this action is non-local, due to the non-locality of the horizon function. However, as we shall see in the next section, it is possible to localize this action by the introduction of a suitable set of auxiliary fields. Here,  differences with respect to the case of the Landau gauge will show up, due to the unavoidable presence of the longitudinal component $A^{aL}_\mu$ of the gauge field.

\section{Localization of the Gribov-Zwanziger action in linear covariant gauges} \label{localization}

In order to have a suitable framework to apply the usual tools of quantum field theory, we have to express the action (\ref{hor3}) in local form. In the case of linear covariant gauges, the localization is not as direct as in Landau gauge. The difficulty relies on  the fact that the horizon function (\ref{hor4}) has two kinds of non-localities. The first one is the same as in Landau gauge, \textit{i.e.} the presence of the inverse of the Faddeev-Popov operator  $\EuScript{M}^T$.  The other one follows  from the fact that the decomposition of the  gauge field into transverse and longitudinal components is non-local, see eq.(\ref{A4}). Therefore, if we apply the same procedure used in the Landau gauge, the localization of the horizon function would give rise to a term of the type

\begin{equation}
\int d^4x~g\gamma^2f^{abc}A^{Tc}_{\mu}(\varphi+\bar{\varphi})^{ab}_{\mu}\,,
\label{hor5}
\end{equation}

\noindent which is still a non-local term, due to the presence of the transverse component  $A^{aT}_\mu$. However, it is possible to localize the action (\ref{hor3}) using an   additional step. First, let us write the transverse component  of the gauge field as

\begin{equation}
A^{Ta}_{\mu} = A^{a}_{\mu} - h^{a}_{\mu}\,,
\label{hor6}
\end{equation}

\noindent where the  field $h^a_\mu$ will be identified with the longitudinal component, {\it i.e.} we shall impose that 

\begin{equation}
h^{a}_{\mu} = \frac{\partial_{\mu}\partial_{\nu}}{\partial^2}A^{a}_{\nu}\,.
\label{hor7}
\end{equation}

\noindent In other words, we introduce an extra field $h^a_\mu$ and state that the transverse part of the gauge field can be written in a local way using eq.(\ref{hor6}). Clearly, we must impose a constraint to ensure that, on-shell, eq.(\ref{hor6}) is equivalent to the usual decomposition (\ref{A4}). Before introducing this constraint, we rewrite the horizon function in terms of $h^a_\mu$. As a matter of notation, we will denote $\EuScript{M}^{T}$ as $\EuScript{M}(A-h)$, when the transverse gauge field is expressed in terms of $h^a_\mu$. The horizon function \eqref{hor4} is now written as  

\begin{equation}
H(A,h) = g^2 \int d^4x~f^{adc}(A^{c}_{\mu}-h^{c}_{\mu})[\left(\EuScript{M}(A-h)\right)^{-1}]^{ab}f^{bde}(A^{e}_{\mu}-h^{e}_{\mu})\,.
\label{hor8}
\end{equation} 

\noindent The constraint given by eq.(\ref{hor7}) is imposed by  means of   a Lagrange multiplier $\lambda^{a}_{\mu}$, \textit{i.e.}, by the introduction of the term

\begin{equation}
S_{\lambda}=\int d^4x~\lambda^{a}_{\mu}(\partial^2h^{a}_{\mu}-\partial_{\mu}\partial_{\nu}A^{a}_{\nu})\,.
\label{hor9}
\end{equation}

\noindent Also, a second constraint is required in order to ensure that the field $(A^{a}_{\mu}-h^{a}_{\mu})$ is transverse, 

\begin{equation}
\partial_{\mu}(A^{a}_{\mu}-h^{a}_{\mu})=0\,,
\label{hor10}
\end{equation}

\noindent which is implemented by the following term

\begin{equation}
S_{\upsilon}=\int d^4x~\upsilon^a\partial_{\mu}(A^{a}_{\mu}-h^{a}_{\mu})\,, 
\label{hor11}
\end{equation}
with $\upsilon^a$ being a Lagrange multiplier. Therefore, the introduction of the extra field $h^a_\mu$ in eq.(\ref{hor8})  by means of  the constraints (\ref{hor7}) and (\ref{hor10}), implemented by the terms (\ref{hor9}) and (\ref{hor11}), provides an action $S'_{\mathrm{GZ}}$ 

\begin{equation}
S'_{\mathrm{GZ}}=S_{\mathrm{YM}}+S_{\mathrm{gf}}+ S_{\lambda}+S_{\upsilon} + \gamma^4H(A,h) - 4V\gamma^4(N^2-1)
\label{hor12}
\end{equation}

\noindent which is on-shell equivalent to the the non-local Gribov-Zwanziger action (\ref{hor3}). The introduction of the fields $h^{a}_{\mu}$, $\lambda^{a}_{\mu}$ and $\upsilon^a$ has to be done through a BRST\footnote{For the usual BRST invariance of the Faddeev-Popov action in linear covariant gauges we remind the reader to Appendix \ref{conventions}.}
 doublet \cite{Piguet:1995er} to avoid the appearance of such fields in the non-trivial part of the cohomology of the BRST operator $s$, a property which will be  important for the renormalizability analysis, \cite{Capri}. Therefore, we introduce the BRST doublets $(h^{a}_{\mu},\xi^{a}_{\mu})$, $(\bar{\lambda}^{a}_{\mu},\lambda^{a}_{\mu})$ and $(\tau^a,\upsilon^a)$, \textit{i.e.}

\begin{alignat}{2}
  sh^{a}_{\mu} &= \xi^{a}_{\mu}\,,\,\,\,\, &&s\lambda^a_{\mu} =0 \nonumber  \\
  s\xi^a_{\mu} &=0\,,\,\,\,\, &&s\tau^a = \upsilon^a \nonumber \\
	s\bar{\lambda}^a_{\mu} &= \lambda^a_{\mu}\,,\,\,\,\, &&s\upsilon^a=0 
\label{hor13}	
\end{alignat}

\noindent and define the following BRST exact  terms

\begin{eqnarray}
S_{\bar{\lambda}\lambda} &=& s\int d^4x~\bar{\lambda}^{a}_{\mu}(\partial^2h^{a}_{\mu}-\partial_{\mu}\partial_{\nu}A^{a}_{\nu})=\int d^4x~\lambda^a_{\mu}(\partial^2h^{a}_{\mu}-\partial_{\mu}\partial_{\nu}A^{a}_{\nu})\,,  \label{nt} \\
&-&\int d^4x~\bar{\lambda}^{a}_{\mu}(\partial^2\xi^a_{\mu}+\partial_{\mu}\partial_{\nu}D^{ab}_{\nu}c^b) \nonumber \\
S_{\tau\upsilon} &=& s\int d^4x~\tau^a\partial_{\mu}(A^{a}_{\mu}-h^a_{\mu}) = \int d^4x~\upsilon^a\partial_{\mu}(A^{a}_{\mu}-h^{a}_{\mu}) \nonumber \\
&+& \int d^4x~\tau^a\partial_{\mu}(D^{ab}_{\mu}c^b+\xi^a_{\mu})\,.
\label{hor14}
\end{eqnarray}

\noindent The terms \eqref{nt} and (\ref{hor14}) implement the constraints (\ref{hor7}) and (\ref{hor10}) in a manifest BRST invariant way. What remains now is the localization of the term (\ref{hor8}). Since this term has just the usual non-locality of the Gribov-Zwanziger action in the Landau gauge, given by the inverse of the operator $\EuScript{M}(A-h)$, we can localize it by the introduction of the same set of auxiliary fields employed in the case of the localization of the horizon function in the Landau gauge \cite{Vandersickel:2012tz}. Thus,  the term  (\ref{hor8}) is replaced by the the following local expression

\begin{eqnarray}
S_H &=& - s\int d^4x~\bar{\omega}^{ac}_{\mu}\EuScript{M}^{ab}(A-h)\varphi^{bc}_{\mu} + \gamma^2g\int d^4x~f^{abc}(A^{a}_{\mu}-h^{a}_{\mu})(\varphi+\bar{\varphi})^{bc}_{\mu} \nonumber \\
&=& \int d^4x~\left(\bar{\varphi}^{ac}_{\mu}\partial_{\nu}D^{ab}_{\nu}\varphi^{bc}_{\mu}-\bar{\omega}^{ac}_{\mu}\partial_{\nu}D^{ab}_{\nu}\omega^{bc}_{\mu}-gf^{adb}(\partial_{\nu}\bar{\omega}^{ac}_{\mu})(D^{de}_{\nu}c^{e})\varphi^{bc}_{\mu}\right) \nonumber \\
&+& \int d^4x~\left(gf^{adb}(\partial_{\nu}\bar{\varphi}^{ac}_{\mu})h^{d}_{\nu}\varphi^{bc}_{\mu}-gf^{adb}(\partial_{\nu}\bar{\omega}^{ac}_{\mu})h^{d}_{\nu}\omega^{bc}_{\mu}-gf^{adb}(\partial_{\nu}\bar{\omega}^{ac}_{\mu})\xi^{d}_{\nu}\varphi^{bc}_{\mu}\right).\nonumber \\
&+& \gamma^2g\int d^4x~f^{abc}(A^{a}_{\mu}-h^{a}_{\mu})(\varphi+\bar{\varphi})^{bc}_{\mu}\,, \nonumber \\
\label{hor15}
\end{eqnarray}

\noindent where

\begin{alignat}{2}
  s\varphi^{ab}_{\mu} &= \omega^{ab}_{\mu}\,,\,\,\,\, &&s\bar{\omega}^{ab}_{\mu} =\bar{\varphi}^{ab}_{\mu} \nonumber  \\
  s\omega^{ab}_{\mu} &=0\,,\,\,\,\, &&s\bar{\varphi}^{ab}_{\mu} = 0\,.  
\label{hor15A}	
\end{alignat}

\noindent Finally,  the action $S_{\mathrm{GZ}}$ given by

\begin{equation}
S_{\mathrm{GZ}}=S_{\mathrm{YM}}+S_{\mathrm{gf}}+S_{\bar{\lambda}\lambda}+S_{\tau\upsilon}+S_H
\label{hor16}
\end{equation}

\noindent is local and, on-shell, equivalent to the non-local action (\ref{hor4}). Explicitly, $S_{\mathrm{GZ}}$ is written as

\begin{eqnarray}
S_{\mathrm{GZ}} &=& \frac{1}{4}\int d^4x~F^{a}_{\mu\nu}F^{a}_{\mu\nu} + \int d^4x~b^a\left(\partial_{\mu}A^{a}_{\mu}-\frac{\alpha}{2}b^a\right)+\int d^4x~\bar{c}^a\partial_{\mu}D^{ab}_{\mu}c^b  \nonumber \\
&+& \int d^4x~\lambda^a_{\mu}(\partial^2h^{a}_{\mu}-\partial_{\mu}\partial_{\nu}A^{a}_{\nu}) +\int d^4x~\bar{\lambda}^{a}_{\mu}(\partial^2\xi^a_{\mu}+\partial_{\mu}\partial_{\nu}D^{ab}_{\nu}c^b)\nonumber \\
&+& \int d^4x~\upsilon^a\partial_{\mu}(A^{a}_{\mu}-h^{a}_{\mu}) + \int d^4x~\tau^a\partial_{\mu}(D^{ab}_{\mu}c^b+\xi^a_{\mu}) + \int d^4x~\left(\bar{\varphi}^{ac}_{\mu}\partial_{\nu}D^{ab}_{\nu}\varphi^{bc}_{\mu} \right.\nonumber \\ 
&-&\left.\bar{\omega}^{ac}_{\mu}\partial_{\nu}D^{ab}_{\nu}\omega^{bc}_{\mu}-gf^{adb}(\partial_{\nu}\bar{\omega}^{ac}_{\mu})(D^{de}_{\nu}c^{e})\varphi^{bc}_{\mu}\right)+\int d^4x~\left(gf^{adb}(\partial_{\nu}\bar{\varphi}^{ac}_{\mu})h^{d}_{\nu}\varphi^{bc}_{\mu}\right.\nonumber \\
&-&\left.gf^{adb}(\partial_{\nu}\bar{\omega}^{ac}_{\mu})h^{d}_{\nu}\omega^{bc}_{\mu}-gf^{adb}(\partial_{\nu}\bar{\omega}^{ac}_{\mu})\xi^{d}_{\nu}\varphi^{bc}_{\mu}\right) + \gamma^2g\int d^4x~f^{abc}(A^{a}_{\mu}-h^{a}_{\mu})(\varphi+\bar{\varphi})^{bc}_{\mu}  \nonumber \\
&-& 4V\gamma^4(N^2-1)  \;, 
\label{hor17}
\end{eqnarray}

\noindent and we will refer to it as the local Gribov-Zwanziger action in linear covariant gauges. We underline  that, in the limit $\gamma\rightarrow 0$, the term (\ref{hor15}) can be trivially integrated out to give a unity.  The remaining action is simply the gauge fixed Yang-Mills action with the addition of the constraints over $h^a_\mu$. These constraints are also easily integrated out, so that the resulting action is simply the usual Yang-Mills action in linear covariant gauges, see Appendix \ref{conventions}.

\noindent  Let us end this section by noticing that, in the local formulation, the gap equation \eqref{gapequation} takes the following expression 

\begin{equation} 
\frac{ \partial {\cal E}_\gamma}{\partial \gamma^2} = 0  \;, \label{locgap}   
\end{equation}

\noindent where ${\cal E}_\gamma$ denotes the vacuum energy, obtained from 

\begin{equation} 
\mathrm{e}^{-V{\cal E}_\gamma} =  
\int  \left[\EuScript{D}\mathbf{\Phi}\right] \mathrm{e}^{-S_{\mathrm{GZ}}}\,. \label{ve}     
\end{equation} 

\noindent with $\Phi$ being the complete set of fields, {\it i.e.}, the usual ones from the Faddeev-Popov quantization and the auxiliary fields introduced to implement the constraints and to localize the Gribov-Zwanziger action.

\section{BRST soft breaking of the Gribov-Zwanziger action in linear covariant gauges} \label{brstbreaking}

As it happens in the case of the Gribov-Zwanziger action in the Landau and maximal Abelian gauges \cite{Vandersickel:2012tz,Dudal:2008sp,Capri:2005tj,Dudal:2006ib,Capri:2006cz,Capri:2008vk,Capri:2010an,Pereira:2013aza,Pereira:2014apa,Dudal:2009xh,Sorella:2009vt,Baulieu:2008fy,Capri:2010hb,Dudal:2012sb,Dudal:2014rxa},
the expression (\ref{hor17}) is not invariant under the BRST transformations (\ref{A3}), (\ref{hor13}) and (\ref{hor15A}). The only term of the action which is not invariant under BRST transformations is

\begin{equation}
g\gamma^2\int d^4x~f^{abc}(A^{a}_{\mu}-h^{a}_{\mu})(\varphi+\bar{\varphi})^{bc}_{\mu}\,,
\label{brstbreaking1}
\end{equation}

\noindent giving

\begin{equation}
sS_{\mathrm{GZ}}\equiv \Delta_{\gamma^2} = g\gamma^2\int d^4x~f^{abc}\left[-(D^{ad}_{\mu}c^d+\xi^a_{\mu})(\varphi+\bar{\varphi})^{bc}_{\mu}+(A^{a}_{\mu}-h^{a}_{\mu})\omega^{bc}_{\mu}\right]\,.
\label{brstbreaking2}
\end{equation}

\noindent  From eq.(\ref{brstbreaking2}), we see that the BRST breaking is soft, {\it i.e.} it is of dimension two in the quantum fields.  This is precisely the same situation of  the Landau and maximal Abelian gauges. The restriction of the domain of integration in the path integral to the Gribov region generates a soft breaking of the BRST symmetry which turns out to be  proportional to the parameter $\gamma^2$. As discussed before, when we take the limit $\gamma\rightarrow 0$, we obtain the usual Faddeev-Popov gauge fixed Yang-Mills action which is BRST invariant. Thus, the breaking of the BRST symmetry is a direct consequence of the restriction of the path integral to the Gribov region. We also emphasize that, although the gauge condition we are dealing with contains  a gauge parameter $\alpha$, the BRST breaking term does not depend on such parameter, due to the fact that the horizon function  takes into account only the transverse component  of the gauge fields, as eq.(\ref{hor4}) shows. 

\section{Analysis of the gap equation at one-loop order} \label{oneloopgap}

As discussed above, in the construction of the effective action which takes into account the presence of infinitesimal Gribov copies a non-perturbative  parameter $\gamma$, {\it i.e.} the Gribov parameter, shows up in the theory. However, this parameter is not free, being determined by the gap equation \eqref{locgap}. In the Landau gauge, the Gribov parameter encodes the restriction of the domain of integration to the Gribov region $\Omega$.  Also, physical quantities like the glueball masses were computed in the Landau gauge \cite{Dudal:2010cd,Dudal:2013wja}, exhibiting an explicit dependence from $\gamma$. It is therefore of  primary importance to look at the Gribov parameter  in linear covariant gauges, where both  the longitudinal component of the gauge field and the gauge parameter $\alpha$ are present in the explicit loop computations.  We should check out  the possible (in)dependence of $\gamma$ from $\alpha$. Intuitively, from our construction, we would expect that the Gribov parameter would be independent from $\alpha$, as a consequence of the fact that we are imposing a restriction of the domain of integration in the path integral which affects essentially only the transverse sector of the theory. Moreover, the independence from $\alpha$ of the Gribov parameter would also imply that  physical quantities like the glueball masses would, as expected,  be $\alpha$-independent. To obtain some computational confirmation of the possible $\alpha$-independence of the Gribov parameter, we provide here the explicit computation of the gap equation at one-loop order. 

\noindent According to eqs.\eqref{locgap},\eqref{ve},  the one-loop vacuum energy can be  computed by retaining the quadratic part of $S_{\mathrm{GZ}}$ and integrating over the auxiliary fields, being given by 

\begin{equation}
\mathrm{e}^{-V{\cal E}^{(1)}_{\gamma}}=\int \left[\EuScript{D}A\right]\mathrm{e}^{-\int \frac{d^4p}{(2\pi)^4}~\frac{1}{2}A^{a}_{\mu}(p)\tilde{\Delta}^{ab}_{\mu\nu}A^{b}_{\nu}(-p)+4V\gamma^4(N^2-1)}\,,
\label{gapeq1}
\end{equation}

\noindent where 

\begin{equation}
\tilde{\Delta}^{ab}_{\mu\nu}=\delta^{ab}\left[\delta_{\mu\nu}\left(p^{2}+\frac{2Ng^2\gamma^4}{p^2}\right)+p_{\mu}p_{\nu}\left(\left(\frac{1-\alpha}{\alpha}\right)-\frac{2Ng^2\gamma^4}{p^4}\right)\right]\,.   \; 
\label{gapeq2}
\end{equation}

\noindent Performing the functional integral over the gauge fields, we obtain 

\begin{equation}
V{\cal E}^{(1)}_{\gamma} = \frac{1}{2}\mathrm{Tr~ln}\tilde{\Delta}^{ab}_{\mu\nu}-4V\gamma^4(N^2-1)\,.
\label{gapeq3}
\end{equation}

\noindent The remaining step is to compute the functional trace in eq.(\ref{gapeq3}). This is a standard computation, see \cite{Vandersickel:2012tz}. Using dimensional regularization, we have 

\begin{equation}
{\cal E}^{(1)}_{\gamma}=\frac{(N^2-1)(d-1)}{2}\int\frac{d^dp}{(2\pi)^d}\mathrm{ln}\left(p^2 + \frac{2Ng^2\gamma^4}{p^2}\right)-d\gamma^4(N^2-1)\,,
\label{gapeq4}
\end{equation}

\noindent where $d= 4-\epsilon$.  We see thus from eq.(\ref{gapeq4}) that the one-loop vacuum energy does not depend on $\alpha$ and the gap equation which determines  the Gribov parameter is written as

\begin{equation}
\frac{\partial {\cal E}^{(1)}_{\gamma}}{\partial \gamma^2} = 0\,\,\, \Rightarrow \,\,\, \frac{(d-1)Ng^2}{d}\int\frac{d^dp}{(2\pi)^d}\frac{1}{p^4+2Ng^2\gamma^4} = 1\,.
\label{gapeq5}
\end{equation}

\noindent This equation states that, at one-loop order, the Gribov parameter $\gamma$ is independent of $\alpha$ and, therefore, is the same as in the Landau gauge, which agrees with our expectation. Although being a useful check of our framework, it is important to state that this result has to be extended at higher orders, a non-trivial topic which is already under investigation  \cite{Capri}.

\section{Dynamical generation of dimension two condensates}\label{condensatessect}

The inclusion of dimension two condensates, see \cite{Dudal:2008sp,Dudal:2011gd}, is an important aspect of the  Gribov-Zwanziger set up in the Landau gauge\footnote{It is worth underlining that these condensates are also present in the maximal Abelian and Coulomb gauges, see  \cite{Capri:2005tj,Dudal:2006ib,Capri:2006cz,Capri:2008vk,Capri:2010an,Guimaraes:2015bra}.}, giving rise to the so-called RGZ action whose predictions are in very good agreement with the most recent lattice numerical data. We underline  that the inclusion of such condensates is not done by hand. They emerge as non-perrurbative dynamical effects due to a quantum instability of the Gribov-Zwanziger theory. In this sense, the RGZ action  \cite{Dudal:2008sp,Dudal:2011gd} takes into account the existence of these condensates in an effective way already at the level of the starting action. Here, we expect that, in analogy with the Landau gauge, dimension two condensates will show up in a similar way. In fact, the presence of these dimension two condensates can be established in a very simple way through a one-loop elementary computation, which shows that the following dimension two condensates 

\begin{equation}
\langle A^{Ta}_{\mu}A^{Ta}_{\mu}\rangle\;,  \qquad \langle \bar{\varphi}^{ab}_{\mu}\varphi^{ab}_{\mu}-\bar{\omega}^{ab}_{\mu}\omega^{ab}_{\mu}\rangle\,.
\label{cond0}
\end{equation} 

\noindent are non-vanishing already at one-loop order, being proportional to the Gribov parameter $\gamma$. In particular, it should be observed that the condensate $\langle A^{Ta}_{\mu}A^{Ta}_{\mu}\rangle$ contains only the transverse component of the gauge field. This is a direct consequence of the fact that the horizon function of the linear covariant gauges, eq.\eqref{hor4}, depends only on the transverse component $A_\mu^{aT}$. In order to evaluate the condensates  $\langle A^T A^T \rangle$ and $\langle \bar{\varphi}\varphi-\bar{\omega}\omega\rangle$  at one-loop order, one needs the quadratic part of the Gribov-Zwanziger action in linear covariant gauges, namely 

\begin{eqnarray}
S^{(2)}_{\mathrm{GZ}} &=& \int d^4x \left[\frac{1}{2}\left(-A^{a}_{\mu}\delta_{\mu\nu}\delta^{ab}\partial^2A^{b}_{\nu}+\left(1-\frac{1}{\alpha}\right)A^{a}_{\mu}\delta^{ab}\partial_{\mu}\partial_{\nu}A^{b}_{\nu}\right) + \bar{c}^{a}\delta^{ab}\partial^2c^{b} \right.\nonumber \\
&+&\left.\bar{\varphi}^{ac}_{\mu}\delta^{ab}\partial^2\varphi^{bc}_{\mu} - \bar{\omega}^{ac}_{\mu}\delta^{ab}\partial^2\omega^{bc}_{\mu}\right]+\gamma^2g\int d^4x~f^{abc}A^{a}_{\nu}\left(\delta_{\mu\nu}-\frac{\partial_{\mu}\partial_{\nu}}{\partial^2}\right)(\varphi + \bar{\varphi})^{bc}_{\mu}\nonumber \\
&-&4V\gamma^4(N^2-1)\,.
\label{cond1}
\end{eqnarray}

\noindent  Further, we  introduce the operators  $\int d^4x A^T A^T $ and $\int d^4x ( \bar{\varphi}\varphi-\bar{\omega}\omega) $ in the action by coupling them to two constant sources $J$ and $m$, and we define the vacuum functional ${\cal E}(m,J)$ defined by 
\begin{equation}
\mathrm{e}^{-V{\cal E}(m,J)}=\int \left[\EuScript{D}\Phi\right]\mathrm{e}^{-S^{(2)}_{\mathrm{GZ}}+J\int d^4x\left(\bar{\varphi}^{ac}_{\mu}\varphi^{ac}_{\mu}-\bar{\omega}^{ac}_{\mu}\omega^{ac}_{\mu}\right)-m\int d^4x~A^{a}_{\mu}\left(\delta_{\mu\nu}-\frac{\partial_{\mu}\partial_{\nu}}{\partial^{2}}\right)A^{a}_{\nu}}\,.
\label{cond2}
\end{equation}

\noindent It is apparent to check that the  condensates $\langle A^T A^T \rangle$ and $\langle \bar{\varphi}\varphi-\bar{\omega}\omega\rangle$ are obtained by differentiating ${\cal E}(m,J)$ with respect to the sources $(J,m)$, which are set to zero at the end, {\it i.e.}

\begin{eqnarray}
\langle \bar{\varphi}^{ac}_{\mu}\varphi^{ac}_{\mu}-\bar{\omega}^{ac}_{\mu}\omega^{ac}_{\mu}\rangle &=& - \frac{\partial {\cal E}(m,J)}{\partial J}\Big|_{J=m=0} \nonumber  \\
\langle A^{Ta}_{\mu}A^{Ta}_{\mu}\rangle &=& \frac{\partial {\cal E}(m,J)}{\partial m}\Big|_{J=m=0}\,.
\label{cond3}
\end{eqnarray}

\noindent A direct computation shows that

\begin{equation}
\mathrm{e}^{-V{\cal E}(m,J)}=\mathrm{e}^{-\frac{1}{2}\mathrm{Tr~ln}\Delta^{ab}_{\mu\nu}+4V\gamma^4(N^2-1)}\,,
\label{cond4}
\end{equation}

\noindent with

\begin{equation}
\Delta^{ab}_{\mu\nu}=\delta^{ab}\left[\delta_{\mu\nu}\left(k^2+\frac{2\gamma^4g^2N}{k^2+J}+2m\right)+k_{\mu}k_{\nu}\left(\left(\frac{1-\alpha}{\alpha}\right)-\frac{2\gamma^4g^2N}{k^2(k^2+J)}-\frac{2m}{k^2}\right)\right]\,. 
\label{cond5}
\end{equation}

\noindent Evaluating the trace, we obtain\footnote{In this computation we are concerned just with the contribution associated to the restriction of the path integral to ${\Omega}_{\mathrm{LCG}}$.}

\begin{equation}
{\cal E}(m,J)=\frac{(d-1)(N^2-1)}{2}\int \frac{d^dk}{(2\pi)^d}~\mathrm{ln}\left(k^2+\frac{2\gamma^4g^2N}{k^2+J}+2m\right)-d\gamma^4(N^2-1)\,. 
\label{cond6}
\end{equation}

\noindent Eq.(\ref{cond3}) and (\ref{cond6}) gives thus

\begin{equation}
\langle \bar{\varphi}^{ac}_{\mu}\varphi^{ac}_{\mu}-\bar{\omega}^{ac}_{\mu}\omega^{ac}_{\mu}\rangle = \gamma^4g^2N(N^2-1)(d-1)\int \frac{d^dk}{(2\pi)^d}\frac{1}{k^2}\frac{1}{(k^4+2g^2\gamma^4N)}
\label{cond7}
\end{equation}

\noindent and

\begin{equation}
\langle A^{Ta}_{\mu}A^{Ta}_{\mu}\rangle = -\gamma^4(N^2-1)(d-1)\int\frac{d^dk}{(2\pi)^d}\frac{1}{k^2}\frac{2g^2N}{(k^4+2g^2\gamma^4N)}\,,
\label{cond8}
\end{equation}

\noindent where we have employed dimensional regularization. Eq.(\ref{cond7}) and eq.(\ref{cond8}) show that, already at one-loop order, both condensates  $\langle A^T A^T \rangle$ and $\langle \bar{\varphi}\varphi-\bar{\omega}\omega\rangle$ are non-vanishing and proportional to the Gribov parameter $\gamma$. Notice also that both integrals in eqs.(\ref{cond7}),(\ref{cond8})  are perfectly convergent in the ultraviolet region by power counting. We see thus that, in perfect analogy with the case of the Landau gauge, dimension two condensates are automatically generated by the restriction of the domain of integration to the Gribov region, as encoded in the  parameter $\gamma$. As shown in \cite{Dudal:2008sp,Dudal:2011gd}, the presence of these condensates can be taken into account directly in the starting action giving rise to the refinement of the Gribov-Zwanziegr action. Also, higher order contributions can be systematically evaluated through the calculation of the effective potential for the corresponding dimension two operators by means of the Local Composite Operator technique, see \cite{Dudal:2011gd,Dudal:2003by}. 

\noindent In the present case, for the refined version of the Gribov-Zwanziger action which takes into account the presence of the dimension two condensates, we get

 \begin{equation}
S_{\mathrm{RGZ}} = S_{\mathrm{GZ}} + \frac{{\hat m}^2}{2}\int d^4x~(A^{a}_{\mu}-h^{a}_{\mu})(A^{a}_{\mu}-h^{a}_{\mu})-{\hat M}^2\int d^4x~(\bar{\varphi}^{ab}_{\mu}\varphi^{ab}_{\mu}-\bar{\omega}^{ab}_{\mu}\omega^{ab}_{\mu})\,,
\label{propa1}
\end{equation}

\noindent where the parameters $({\hat m},{\hat M})$ can be determined order by order in a self-content way through the evaluation of the corresponding effective potential, as outlined in the case of the Landau gauge  \cite{Dudal:2011gd}.  Let us remark here that the calculation of the vacuum functional ${\cal E}(m,J)$, eq. \eqref{cond2}, done in the previous section shows that, at one-loop order, these parameters turn out to be independent from $\alpha$. The study of the effective potential for the dimension two operators  $\int d^4x A^T A^T $ and $\int d^4x ( \bar{\varphi}\varphi-\bar{\omega}\omega) $ will be thus of utmost importance in order to extend this feature to higher orders.

\noindent We are now ready to evaluate the tree level gluon propagator in the linear covariant gauges. This will be the topic of the next section.  

\section{Gluon propagator and comparison with the most recent lattice data}\label{propagator}

From the refined Gribov-Zwanziger action, eq.(\ref{propa1}), one can immediately evaluate the  tree level gluon propagator in linear covariant gauges, given by the following expression 

\begin{equation}
\langle A^{a}_{\mu}(k)A^{b}_{\nu}(-k)\rangle = \delta^{ab}\left[\frac{k^2+{\hat M}^2}{(k^2+{\hat m}^2)(k^2+{\hat M}^2)+2g^2\gamma^4N}\left(\delta_{\mu\nu}-\frac{k_{\mu}k_{\nu}}{k^2}\right)+\frac{\alpha}{k^2}\frac{k_{\mu}k_{\nu}}{k^2}\right]\,.
\label{propa2}
\end{equation}

\noindent A few comments are now in order.  First, the longitudinal sector is not affected by the restriction to the Gribov region, \textit{i.e.} the longitudinal component of the propagator is the same as the perturbative one. This is an expected result,  since the Gribov region ${\Omega}_{\mathrm{LCG}}$ for linear covariant gauges does not impose any restriction to the longitudinal component $A^{aL}_\mu$. Second, in the limit $\alpha\rightarrow 0$, the gluon propagator coincides precisely with the known result in Landau gauge \cite{Dudal:2008sp}. In particular, this  implies that all features  of the RGZ framework derived in Landau gauge remains true for  the transverse component of the correlation function \eqref{propa2}.

\subsection{Lattice results}

As we have already mentioned before, unlike the Landau, Coulomb and maximal Abelian gauge, the linear covariant gauges do not exhibit a minimizing functional, a feature which is the source of several  complications in order to construct a lattice formulation of these gauges. The study of the linear covariant gauges through  lattice numerical simulations represents a big challenge. The first attempt to implement these gauges on the lattice was undertaken by   \cite{Giusti:1996kf,Giusti:1999im,Giusti:2000yc,Giusti:2001kr}. More recently, the authors  \cite{Cucchieri:2008zx,Mendes:2008ux,Cucchieri:2010ku,Cucchieri:2009kk,Cucchieri:2011pp,Cucchieri:2011aa} have been able to implement the linear covariant gauges on the lattice by means of a different procedure. 

\noindent With respect to the most recent data \cite{Cucchieri:2008zx,Mendes:2008ux,Cucchieri:2010ku,Cucchieri:2009kk,Cucchieri:2011pp,Cucchieri:2011aa}  obtained on bigger lattices, our results are in very good qualitative agreement: The tree level transverse gluon propagator does not depend on $\alpha$ and, therefore, behaves like the gluon propagator in the refined Gribov-Zwanziger framework in the Landau gauge. On the other hand, the longitudinal form factor $D_L(k^2)$  defined by 
\begin{equation} 
D_L(k^2) =k^2  \frac{\delta^{ab}}{N^2-1} \frac{k_\mu k_\nu}{k^2} \langle A^a_\mu(k) A^b_\nu(-k) \rangle= \alpha   \;. \label{lff}
\end{equation} 
 is equal to the gauge parameter $\alpha$, being not affected by the restriction to the Gribov region ${\Omega}_{\mathrm{LCG}}$. These results are in complete agreement with the numerical data of 
 \cite{Cucchieri:2008zx,Mendes:2008ux,Cucchieri:2010ku,Cucchieri:2009kk,Cucchieri:2011pp,Cucchieri:2011aa}. Although many properties of the Gribov region ${\Omega}_{\mathrm{LCG}}$ in linear covariant gauges need to be further established, the qualitative agreement of our results on the gluon propagator with the recent numerical are certainly reassuring, providing a good support  for the introduction of the region ${\Omega}_{\mathrm{LCG}}$.

\noindent To end this section we point out that, recently, results for the gluon and ghost propagators in linear covariant gauges have been  obtained through the use of the Dyson-Schwinger equations by \cite{Aguilar:2015nqa,Huber:2015ria}.



\section{Conclusions}

In this work we have presented a proposal for a region ${\Omega}_{\mathrm{LCG}}$ in field space which is free from  infinitesimal Gribov copies in linear covariant gauges. A local implementation of the restriction of the domain of integration to the region ${\Omega}_{\mathrm{LCG}}$ has been worked out, resulting in the refined Gribov-Zwanziger action given in eqs.\eqref{hor17}, \eqref{propa1}. 

\noindent The tree-level gluon propagator stemming from the refined action \eqref{propa1} has been evaluated, being in good qualitative agreement with the most recent numerical lattice simulations   \cite{Cucchieri:2008zx,Mendes:2008ux,Cucchieri:2010ku,Cucchieri:2009kk,Cucchieri:2011pp,Cucchieri:2011aa}. Although many geometrical properties of the region ${\Omega}_{\mathrm{LCG}}$ remain to be established, we regard the agreement with the lattice data as a first encouraging step towards a non-perturbative analytic formulation of  the linear covariant gauges.  

\noindent Also, we have proven that at one-loop order, the Gribov parameter $\gamma$ is independent from the gauge parameter $\alpha$. This result is a kind of consistency check of the viability of our formulation,  since $\gamma$ enters explicitly the computation of physical quantities like the masses of the  glueballs. 

\noindent We  underline that, within the class of the covariant renormalizable gauges,  the non-perturbative treatment of Yang-Mills theories, both at analytical and numerical level, is usually carried out in the Landau gauge, for which a large amount of non-perturbative results has been already achieved. In this sense, the non-perturbative study of the linear covariant gauges is at the beginning. From this perspective, our results can be seen as  a first step in order to  address the Gribov issue in this class of gauges.  We hope that it might be useful for the well succeeded interplay between numerical and analytical analysis. 

\noindent Let us end by giving a partial list of  aspects of the linear covariant gauges which we intend to investigate in the near future.  First, we should study the renormalizability properties of the Gribov-Zwanziger action in linear covariant gauges \cite{Capri}. Second, a better understanding of the geometrical properties of the region ${\Omega}_{\mathrm{LCG}}$ needs to be achieved.  Third, the inclusion of quark matter fields  along the lines outlined in \cite{Capri:2014bsa} can provide information about the quark propagator and its non-perturbative mass function in these gauges. Fourth, a detailed study of the Faddeev-Popov ghost propagator remains to be worked out.  Fifth, the evaluation of the effective potential for the dimension two operators $\int d^4x A^T A^T $ and $\int d^4x ( \bar{\varphi}\varphi-\bar{\omega}\omega) $ will enable us to extend to higher orders the $\alpha$-independence of the massive parameters $(\gamma, {\hat m}, {\hat M})$ entering the refined Gribov-Zwanziger action in the linear covariant gauges. We hope to report soon on those interesting topics.

\section*{Acknowledgements}

The authors are grateful to M.~S.~Guimar\~aes and L.~F.~Palhares for many helpful discussions. ADP acknowledges Markus Huber for enlightening comments on recent results about linear covariant gauges. The Conselho Nacional de Desenvolvimento Cient\'{i}fico e Tecnol\'{o}gico\footnote{ RFS is a level PQ-2 researcher under the program \emph{Produtividade em Pesquisa}, 304924/2009-1.} (CNPq-Brazil), The Coordena\c c\~ao de Aperfei\c coamento de Pessoal de N\'ivel Superior (CAPES) and the Pr\'o-Reitoria de Pesquisa, P\'os-Gradua\c c\~ao e Inova\c c\~ao (PROPPI-UFF) are acknowledge for financial support.


\appendix

\section{Conventions and notation}\label{conventions}

In this paper, we consider pure Yang-Mills theories in four Euclidean dimensions with $SU(N)$ gauge group. We choose the linear covariant gauges as our gauge condition, \textit{i.e.}

\begin{equation}
\partial_{\mu}A^{a}_{\mu}=\alpha b^a\,,
\label{A1}
\end{equation}

\noindent where $A^{a}_{\mu}$ denotes the gauge field, $\alpha$ is a gauge parameter and $b^a$ is the Lautrup-Nakanishi Lagrange multiplier.

\noindent Therefore, the gauge fixed Yang-Mills action is 

\begin{eqnarray}
S_{\mathrm{YM}} +S_{\mathrm{gf}}&=&\frac{1}{4}\int d^4x~F^{a}_{\mu\nu}F^{a}_{\mu\nu}+s\int d^4x~\bar{c}^{a}\left(\partial_{\mu}A^{a}_{\mu}-\frac{\alpha}{2}b^a\right)\nonumber \\
&=& \frac{1}{4}\int d^4x~F^{a}_{\mu\nu}F^{a}_{\mu\nu} + \int d^4x~\left[b^a\left(\partial_{\mu}A^{a}_{\mu}-\frac{\alpha}{2}b^a\right)+\bar{c}^{a}\partial_{\mu}D^{ab}_{\mu}c^b\right]\,,
\label{A2}
\end{eqnarray}

\noindent with $F^{a}_{\mu\nu}=\partial_{\mu}A^{a}_{\nu}-\partial_{\nu}A^{a}_{\mu}+gf^{abc}A^{b}_{\mu}A^{c}_{\nu}$ being the field strength, $\bar{c}^{a}$ and $c^a$ the Faddeev-Popov antighost and ghost respectively, $s$ the nilpotent BRST operator and $D^{ab}_{\mu}=\delta^{ab}\partial_{\mu}-gf^{abc}A^{c}_{\mu}$ the covariant derivative in the adjoint representation. The BRST transformations for the fields $(A,c,\bar{c},b)$ are 

\begin{eqnarray}
sA^{a}_{\mu} &=& -D^{ab}_{\mu}c^b \nonumber \\
sc^a &=& \frac{g}{2}f^{abc}c^bc^c \nonumber \\
s\bar{c}^{a} &=& b^a \nonumber \\
sb^a &=& 0\,.
\label{A3}
\end{eqnarray}

\noindent The gauge condition (\ref{A1}) imposes a constraint on  the longitudinal component of the gauge field, since the divergence of the transverse sector is identically vanishing. Therefore, it is useful to decompose the gauge field into transverse and longitudinal components,

\begin{eqnarray}
A^{Ta}_{\mu}&=&\left(\delta_{\mu\nu}-\frac{\partial_{\mu}\partial_{\nu}}{\partial^2}\right)A^{a}_{\nu} \nonumber \\
A^{La}_{\mu}&=&\frac{\partial_{\mu}\partial_{\nu}}{\partial^2}A^{a}_{\nu}\,.
\label{A4}
\end{eqnarray}

\noindent Clearly, the choice $\alpha=0$ in eq.(\ref{A1}) makes the gauge condition equivalent to the Landau gauge, which implies the vanishing of the longitudinal component of the gauge field. 


\end{document}